\documentclass[12pt]{article}

\usepackage{amssymb} 
\usepackage{amsmath} 
\usepackage{amsfonts} 
\usepackage{eurosym} 

\usepackage{geometry} 
\usepackage{ulem} 
\usepackage{graphicx} 
\usepackage{caption} 
\usepackage{color} 
\usepackage{setspace} 
\usepackage{sectsty} 
\usepackage{comment} 
\usepackage{footmisc} 
\usepackage{pdflscape} 
\usepackage{subfigure} 
\usepackage{array} 
\usepackage{multirow} 

\usepackage{natbib} 
\usepackage{hyperref} 

\usepackage{mathtools} 
\mathtoolsset{showonlyrefs,showmanualtags} 

\usepackage[bibliography=common]{apxproof} 

\newtheorem{theorem}{Theorem}

\newtheorem{prop}{Proposition}
\newtheorem{cor}{Corollary}
\newtheorem{ass}{Assumption}
\newtheorem{remark}{Remark}

\newtheorem{modl}{Model}

\newcolumntype{L}[1]{>{\raggedright\let\newline\\arraybackslash\hspace{0pt}}m{#1}}
\newcolumntype{C}[1]{>{\centering\let\newline\\arraybackslash\hspace{0pt}}m{#1}}
\newcolumntype{R}[1]{>{\raggedleft\let\newline\\arraybackslash\hspace{0pt}}m{#1}}

\newcommand{\E}{\mathbb{E}} 
\newcommand{\ind}{\mathbf{1}} 
\newcommand{\R}{\mathbb{R}} 

\begin{document}

\begin{titlepage}
\title{Manipulation Test for Multidimensional RDD\thanks{I am grateful to Federico Bugni and Ivan Canay for their guidance in this project. I am also thankful to all the participants of the Econometric Reading Group at Northwestern University for their comments and suggestions.}}
\author{Federico Crippa \\ Department of Economics \\ Northwestern University\\
\href{mailto:federicocrippa2025@u.northwestern.edu}{federicocrippa2025@u.northwestern.edu}
}
\date{\today}
\maketitle
\begin{abstract}
\noindent The causal inference model proposed by \cite{lee2008randomized} for the regression discontinuity design (RDD) relies on assumptions that imply the continuity of the density of the assignment (running) variable. The test for this implication is commonly referred to as the manipulation test and is regularly reported in applied research to strengthen the design's validity. The multidimensional RDD (MRDD) extends the RDD to contexts where treatment assignment depends on several running variables. This paper introduces a manipulation test for the MRDD. First, it develops a theoretical model for causal inference with the MRDD, used to derive a testable implication on the conditional marginal densities of the running variables. Then, it constructs the test for the implication based on a quadratic form of a vector of statistics separately computed for each marginal density. Finally, the proposed test is compared with alternative procedures commonly employed in applied research.\\
\vspace{0in}
\vspace{0in}\\

\bigskip
\end{abstract}
\setcounter{page}{0}
\thispagestyle{empty}
\end{titlepage}
\pagebreak \newpage

\doublespacing

\section{Introduction}
Regression Discontinuity Design (RDD) is widely used in policy evaluation and causal inference analysis to establish credible causal relationships under mild assumptions. RDD requires that units are assigned to a treatment based on some observable characteristic, the running variable: the probability of being treated must discontinuously change when the value of the running variable exceeds a certain threshold, called the cutoff. The fact that policies are often designed in this way (scholarship for students with GPA exceeding a threshold, welfare benefits for households with a certain income, etc.) explains RDD popularity\footnote{See \cite{abadie2018econometric}, \cite{cattaneo2019practicalf}, and \cite{cattaneo2019practicale} for recent comprehensive reviews on RDD applications, identification, estimation, and inference.}.

To identify the average treatment effect at the cutoff, the model for causal inference with the RDD proposed by \cite{lee2008randomized} requires assumptions on unobservable potential outcomes. Albeit these assumptions are not directly testable, the model entails two implications on observable quantities, which can be tested.
The first implication requires continuity at the cutoff for the probability density function of the running variable. The second imposes continuity at the cutoff for the conditional (on the running variable) expectation of additional observable characteristics measured before the treatment. Tests of these implications provide evidence of the RDD validity and are commonly reported in empirical applications, as highlighted by the survey in \cite{canay2018approximate}. The one for the continuity of the density of the running variable is known as the manipulation test since it checks that units do not manipulate their scores to get assigned to treatment. Several manipulation tests have been proposed in the literature, from the seminal one by \cite{mccrary2008manipulation} to more recent approaches by \cite{cattaneo2020simple} and \cite{bugni2021testing}.

This paper introduces a manipulation test for the multidimensional RDD (MRDD), valid for both cases of perfect (sharp MRDD) and imperfect (fuzzy MRDD) compliance. MRDD is a model where the treatment assignment depends on multiple running variables. I consider the version of MRDD where the probability of receiving the treatment changes discontinuously when all the running variables exceed their cutoffs, and the cutoff of each running variable is fixed\footnote{This shape of the assignment region is the most popular in practice (see references below), but may exclude the spatial RDD.}. Compared to the single-dimensional RDD, the main novelty is that the cutoff is not a point in a single-dimensional space but a set of infinite points in the multidimensional space of the running variables. Consider, for example, a scholarship for students who score above certain thresholds in language and mathematics tests, as illustrated in Figure \ref{fig:mrddintro}. The MRDD allows the researcher to identify and estimate the average effect of the scholarship on students with scores at the boundary. In this case, the cutoff is the solid purple boundary in the bi-dimensional space of language and math scores.

\begin{figure}
\begin{center}
\resizebox{0.8\textwidth}{!}{
\includegraphics{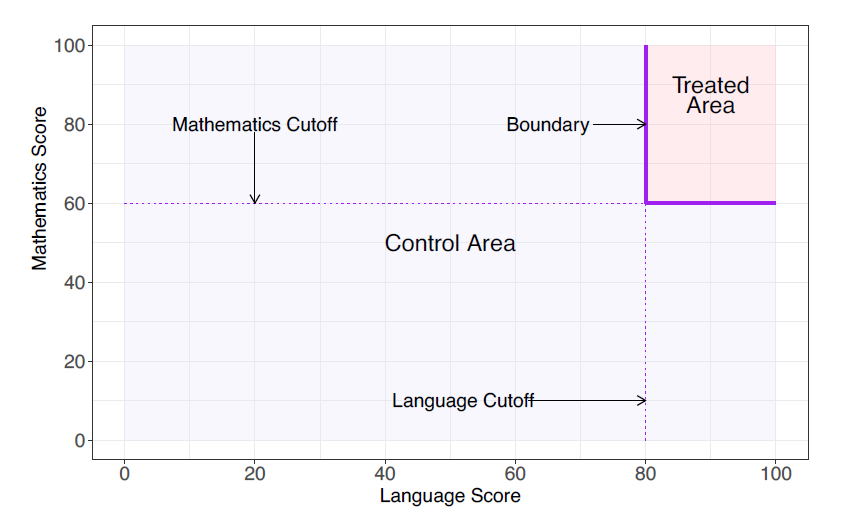}
}
\caption{MRDD running variables and thresholds. Source: \cite{cattaneo2019practicale}}
\label{fig:mrddintro}
\end{center}
\end{figure}

The main contributions of this paper are the following: first, I extend the \cite{lee2008randomized}'s model to the multidimensional setting and derive a testable implication on the conditional marginal densities of the running variables. Then, I construct a manipulation test for the implication, which helps corroborate the MRDD's credibility. Intuitively, the test procedure first divides the space of running variables into subspaces where only one variable determines the treatment assignment. In each subspace, the model implication requires the marginal density of the now single-running variable to be continuous at its single-dimensional threshold. I then obtain a set of conditions on the continuity of conditional marginal densities of all running variables and propose a multidimensional manipulation test based on a quadratic form of the test statistics considered by \cite{cattaneo2020simple} for the single-dimensional case, computed for each implication. Asymptotically, these statistics converge to a multivariate normal distribution, and my test statistic to a chi-square with as many degrees of freedom as the number of running variables. Finally, I compare my procedure with alternative methods, providing a review of other approaches used in the literature.

I am not the first to study MRDD from a theoretical perspective. Identification and estimation in the MRDD setting, and how they differ from the single-dimensional RDD, have been investigated by \cite{imbens2009regression} and \cite{papay2011extending}, respectively; other results are also discussed in \cite{wong2013analyzing}, and \cite{imbens2019optimized}. So far, to my knowledge, there is no research explicitly dealing with extending the framework proposed by \cite{lee2008randomized} and discussing manipulation tests in the MRDD context. Interestingly, though, manipulation tests are run by several applied papers employing MRDD (see the survey in Table \ref{tab:survey}): they appeal to disparate approaches, with different null hypotheses, assumptions, and test statistics. None of these approaches justifies the implemented procedures\footnote{\cite{snider2015barriers} recognize that formal results are missing, asserting that \textit{Extending formal tests to check for the strategic manipulation [\dots] with a two-dimensional predictor vector is not immediately clear}.}, while my test is supported by a model and backed by statistical theory. Local asymptotic analysis and Monte Carlo simulations confirm my test's advantages in terms of size control and power in realistic settings.

Three main strands of literature resort to MRDD as a tool for causal inference. First, it is exploited to evaluate policies that assign a treatment when more than one condition on observable continuous quantities is met. Examples can be found in several fields, mainly in education \citep{matsudaira2008mandatory,clark2014signaling,cohodes2014merit,elacqua2016short,evans2017smart,smith2017giving,londono2020upstream} but also in corporate finance \citep{becht2016does}, political economy \citep{hinnerich2014democracy,frey2019cash}, development \citep{salti2022impact}, industrial organization \citep{snider2015barriers}, and public economics \citep{egger2015impact}. In these cases, MRDD provides reliable results on treatment effects from a clean identification strategy.

A second application is the geographic or spatial RDD to study the effect of treatments only assigned to specific areas. Running variables are latitude and longitude, and the boundary at which the ATE is computed coincides with actual (or historical) national, regional, or municipal borders. \cite{keele2015geographic} discuss how this setting relates to MRDD in detail. Note, however, that I am considering a model where the treatment is assigned when each running variable exceeds its cutoff: as such, my results do not directly apply to the spatial RDD, and if my test works in this setting, it is a case-specific issue.

Third, recently there has been an increasing interest in MRDD in a theory literature at the intersection of market design and machine learning \citep{abdulkadiroglu2022breaking,narita2021algorithm}. When algorithms determine a treatment assignment, they may consider multiple thresholds and running variables in a setting that mimics an MRDD. This literature is primarily theoretical, but it will likely encourage new empirical research, potentially relying on my proposed manipulation test.

The rest of the paper is organized as follows. Section \ref{sec:model} introduces the theoretical model for MRDD and derives the testable implication. Section \ref{sec:test} provides a manipulation test for the implication. Section \ref{sec:comparison} compares the manipulation test with alternative approaches used in the literature. Section \ref{sec:mc} reports Monte Carlo simulations. Section \ref{sec:app} applies the manipulation test to \cite{frey2019cash}. Section \ref{sec:fut} concludes.

\section{Model and Testable Implication}\label{sec:model}

Let $Z\in \R^d$ be a random vector of $d$ observable continuous running variables, with joint cumulative distribution function $F(z)$ and joint probability distribution function $f(z)$.

The treatment status $D$ depends on $Z$. This paper considers a sharp design where the treatment status is deterministic, but the setting can also extend to a fuzzy design where the probability of being treated changes discontinuously at the threshold. Units receive the treatment ($D=1$) when all the running variables exceed their cutoff:
\begin{gather} \label{eq:treatment_assignment}
D = D(Z) = \ind \{Z \geq c \} = \ind \{Z_1 \geq c_1 \}\ind \{Z_2 \geq c_2 \}\dots\ind \{Z_d \geq c_d \}.
\end{gather}
Without loss of generality, I consider a rescaling of $Z$ such that $c_j=0$ for all $j$. Units are treated when $Z_j \geq 0$ for all $j$.

Let $\mathcal{T}$ be the set of values of $Z$ for which $D(Z)=1$, and let $\bar{\mathcal{T}^C}$ denote the closure of the complement of $\mathcal{T}$. Define the boundary $\mathcal{B}$, the set of points with both treated and untreated units in any neighborhood, as $\mathcal{B} = \mathcal{T} \cap \bar{\mathcal{T}^C}$.

The MRDD exploits the treatment assignment mechanism described in Equation \eqref{eq:treatment_assignment} to identify the causal effect of the treatment at the boundary, where treated and untreated units are comparable and differ only in their treatment status. \citet{imbens2009regression} show that, for all $b \in \mathcal{B}$, the Conditional Average Treatment Effects $\tau(b) = \E[Y_1 -Y_0 |Z=b]$ is identified ($Y_0$ and $Y_1$ are the potential outcomes, in a standard notation). 

In Appendix \ref{app:model}, I extend the model proposed by \citet{lee2008randomized} for the RDD to the multidimensional setting and show how the assumptions of the MRDD on the unobservable potential outcomes lead to an implication on the observable distribution of $Z$. This implication is presented in the following proposition.

\begin{prop} \label{prop:impl}
{\normalfont (Testable implication)}
    To identify the causal parameter $\tau(b)$, the MRDD relies on assumptions with the following implication:
    \begin{gather}
    f_{Z_j|Z_{-j}}(z_j|z_{-j}\geq0) \text{ continuous at } z_j=0, \forall j \label{eq:condition_marginal}
\end{gather}
    where $f_{Z_j|Z_{-j}}$ is the conditional probability density function of the running variable $Z_j$, conditional on all the other running variables $Z_{-j} = (Z_1, \dots, Z_{j-1},Z_{j+1}, \dots, Z_d)$.
\end{prop}


The implication in Equation \eqref{eq:condition_marginal} requires the continuity at the boundary of the conditional marginal density of each running variable. The condition is equivalent to the continuity of the marginal density of each random variable $j$ at its own threshold ($z_j =0$), when all the other running variables are above their cutoffs ($z_{-j} \geq 0)$. In terms of restrictions on agents' behavior, the implication has a clear interpretation: for any running variable, it excludes manipulation that may switch the treatment status. 

In the next section, I propose a manipulation test for the implication in Equation \eqref{eq:condition_marginal}. It is important to stress that the implication is necessary but not sufficient: the assumptions of the MRDD model may be violated even if the condition is satisfied. This issue does not depend on $d$ being larger than one and similarly arises in the single-dimensional case (see \citet{lee2008randomized} for further discussion). The manipulation test should be viewed as a robustness check and in any specific application cannot substitute a discussion on the validity of the assumptions of the MRDD.

\begin{remark}
{\normalfont (Alternative Testable Implication)}
From the assumptions of the MRDD, it is possible to derive alternative testable implications on the density of the running variables. For example, they also imply the continuity of the joint density $f(Z)$ at the boundary. In the context where the treatment is assigned as described in Equation \eqref{eq:treatment_assignment}, a condition on the joint density is harder to interpret than the one on the marginals. Moreover, the fact that it concerns continuity at infinite points (the boundary) seems to be a practical disadvantages, as testing the condition of a multidimensional density in an infinite number of points may determine low power in many applications. In other contexts, however, like the geographic RDD, the condition on the joint distribution has a clear interpretation and can constitute the proper null hypothesis for a manipulation test. Proposing a test valid also for the geographic RDD goes beyond the scope of this paper, and how to test the continuity of a multivariate density remains an interesting open problem.
\end{remark}

\section{Manipulation Test} \label{sec:test}

In this section, I first introduce the manipulation test for the implication in Equation \eqref{eq:condition_marginal}, presenting the test statistic and the critical values. Practitioners interested in applying the manipulation test will find a practical description of how to implement the procedure. Then, in Section \ref{subsec:fhat}, I discuss the assumptions necessary to establish the asymptotic validity and consistency of the test. Finally, in Section \ref{subsec:manipulation_test}, I derive the formal results.

The MRDD manipulation test I propose for the implication in Equation \eqref{eq:condition_marginal} is defined as follows:
\begin{gather} \label{eq:manipulation_test_phi}
\phi(\hat t,\alpha) = \begin{cases}
1, & \text{if $\hat t>c(\alpha)$}\\
0, & \text{if $\hat t\leq c(\alpha)$}
\end{cases}
\end{gather}
where $\hat t$ is the test statistic, $\alpha$ the significance level, and $c(\alpha)$ the critical value. Whenever $\phi(\hat t, \alpha)$ equals 1, the null hypothesis is rejected.

The construction of the test statistic $\hat{t}$ involves two steps. First, for each running variable $j$, compute the statistic $\hat{\theta}_j$ along with its variance estimator $\hat{\sigma}^2_j$. The expression for $\hat{\theta}_j$ is given by:
\begin{gather}
    \hat{\theta}_j = \hat{f}^{+}_{Z_j | Z_{-j}}(0|z_{-j} \geq 0) - \hat{f}^{-}_{Z_j | Z_{-j}}(0|z_{-j} \geq 0)
\end{gather}
where $\hat{f}^{+}_{Z_j | Z_{-j}}$ and $\hat{f}^{-}_{Z_j | Z_{-j}}$ are estimators of the conditional marginal density of $Z_j$, as defined in Section \ref{subsec:fhat}. It is worth noting that $\hat{\theta}_j$ resembles the test statistic proposed by \cite{cattaneo2020simple} for testing the continuity of the density in a single-dimensional RDD. However, in this context, the test focuses on the continuity of a \textit{conditional} marginal density, which requires adaptations for the statistic and the formal proofs.

In practice, $\hat{\theta}_j$ and $\hat{\sigma}_j$ can be computed as follows: for any $j=1, \dots, d$, consider the subsample of data that includes only observations where $Z_{-j} \geq 0$. Within this subsample, use available packages in R and Stata to run the manipulation test proposed by \cite{cattaneo2020simple} and obtain the test statistic $\hat{\theta}_j$ along with its standard error estimator $\hat{\sigma}_j$.

Next, construct the test statistic $\hat{t}$ based on the vector $\hat{\theta} = (\hat{\theta}_1, \dots, \hat{\theta}_d)$ and the diagonal matrix $\hat{\Sigma} = \text{diag}(\hat{\sigma}^2_1, \dots, \hat{\sigma}^2_d)$, using the quadratic form:
\begin{gather} \label{eq:quadratic_form}
    \hat{t} = \hat{\theta}' \hat{\Sigma}^{-1} \hat{\theta}.
\end{gather}

In Section \ref{subsec:manipulation_test}, I derive the asymptotic distribution of $\hat{\theta}$ and prove that the test statistic $\hat{t}$ converges to a $\chi^2$ distribution with $d$ degrees of freedom. Consequently, the critical value $c(\alpha)$ can be chosen as the $1-\alpha$ quantile of a $\chi^2$ distribution with $d$ degrees of freedom.

With these test statistics and critical value, the manipulation test $\phi(\hat{t}, \alpha)$ is asymptotically valid and consistent for the null hypothesis in Equation \ref{eq:condition_marginal}.

\subsection{Assumptions} \label{subsec:fhat}

To prove the properties of the manipulation test, the following assumption is required.

\begin{ass}\label{ass:cdf}
{\normalfont (Smoothness)} $\{z_i\}_{i \in \{1, \dots, n\}}$ is an iid random sample of $Z$ with cumulative distribution function $F$. In neighborhoods of points on the boundary $\mathcal{B}$, $F$ is at least four times continuously differentiable.
\end{ass}

The assumption that $F$ has at least four continuous derivatives allows for consistent estimation of the conditional densities $f^{+}_{Z_j | Z_{-j}}(0)$ and $f^{-}_{Z_j | Z_{-j}}(0)$. This is analogous to the assumption required by the manipulation test of \cite{mccrary2008manipulation} for the single-dimensional case.

To simplify the notation, define $f_j(z_j) = f_{Z_j|Z_{-j}}(z_j|z_{-j}\geq0)$, and consider the local polynomial estimator $\hat{f}_{j,p}(z_j)$:
\begin{gather*}
\hat{f}_{j,p}(z_j) = e_1'\hat{\beta}(z_j) \\
\hat{\beta}(z_j) = \text{argmin}_{b\in \R^{p+1}} \sum_{i=1}^{n}\left[\tilde{F_j}(z_j|z_{-j}\geq 0) - r_p(z_{ji} -z_j)'b \right]^2 K\left(\frac{z_{ji} -z_j}{h_j}\right) \ind\{z_{-j}\geq 0\}
\end{gather*}
where $e_1'\in \R^{p+1}$ such that $e_1' = (0,1,0,\dots,0)$; $n_j = \sum_{i=1}^{n} \ind\{z_{-j}\geq 0\}$ is the number of observations actually considered for the test; $\tilde{F_j}(z_j|z_{-j}\geq 0) = \frac{1}{n_j} \sum_i \ind \{ z_{ji} \leq z_j\} \ind\{z_{-ji}\geq 0\}$ is the empirical distribution function for the marginal conditional distribution $F(z_j|z_{-j}\geq 0)$; $r_p(u) =(1,u,u^2,\dots,u^p)$ is a one-dimensional polynomial expansion; $h_j$ is a bandwidth, which will be better specify later; and $K(\cdot)$ is a kernel function satisfying the following assumption.

\begin{ass}\label{ass:kernel}
{\normalfont (Kernel)} The kernel function $K(.)$ is nonnegative, symmetric, continuous, and integrates to one: $\int K(u)du =1$. It has support $[-1,1]$.
\end{ass}


The local polynomial approach for estimating derivatives of the cumulative distribution function was introduced by \cite{jones1993simple} and \cite{fan2018local}. This method is particularly well-suited for the manipulation test as it achieves the optimal rate of convergence both at interior points and at the boundaries and is boundary adaptive. No adjustment is required when computing estimates for points near the boundary if the object of interest is the $\nu$-th derivative and $p - \nu$ is odd. For my manipulation test, where $\nu = 1$, I will use an estimator with $p = 2$ to take advantage of its boundary adaptiveness.

\subsection{Manipulation Test} \label{subsec:manipulation_test}

To establish asymptotic validity and consistency of the test $\phi(\hat{t},\alpha)$, it is necessary to derive some intermediate results. The first result regards the asymptotic properties of the density estimator $\hat{f}_{j,p}(z_j)$. Formulas for bias $B(x)$, variance $V(x)$, and consistent variance estimator $\hat{V}(x)$ are reported in Appendix \ref{app:formulas}.

\begin{prop} \label{thm:cons}
{\normalfont (Asymptotic distribution of $\hat{f}_{j,p}(z_j)$)}
Under Assumptions \ref{ass:cdf} and \ref{ass:kernel}, with $p=2$, $n h_j^2 \rightarrow \infty$ and $n h_j^{2p+1}=O(1)$, $\hat{f}_{j,p}(z_j)$ is a consistent estimator for $f_{j}(z_j)$. Furthermore,
\begin{gather}
\sqrt{n_j h_j} (\hat{f}_{j,p}(z_j) - f_j(z_j) - h_j^p B(z_j)) \rightarrow^d \mathcal{N}(0,V(z_j))
\end{gather}
where $B(z_j)$ is the asymptotic bias, and $V(z_j)$ the asymptotic variance.
\end{prop}

When $n h_j^{2p+1} = O(1)$, the bandwidth $h_j$ has the MSE-optimal rate and can be computed by cross-validation.

The presence of asymptotic bias $B(z_j)$ is standard in nonparametric settings and must be considered to ensure valid hypothesis testing. In this paper, I adopt the robust bias correction method proposed by \cite{calonico2018effect}. Alternative approaches include the critical values correction method suggested by \cite{armstrong2020simple}.

Bias-corrected inference for the density estimator $\hat{f}_{j,p}(z_j)$ can be obtained by considering the estimator $\hat{f}_{j,q}(z_j)$ with $q=p+1$, computed with the bandwidth $h_{j,p}$, the MSE-optimal bandwidth for $\hat{f}_{j,p}(z_j)$ (see \cite{calonico2022coverage} and \cite{cattaneo2019lpdensity} for an extensive discussion on the procedure). Moving forward, I will consider the estimator $\hat{f}_{j,p}(z_j)$ for point estimates, and the estimator $\hat{f}_{j,q}(z_j)$ with bandwidth $h_{j,p}$ to construct bias-corrected confidence intervals for $\hat{f}_{j,p}(z_j)$.

Let $n_{j+} = \sum_{i=1}^{n} \ind\{z_{j}\geq 0\}\ind\{z_{-j}\geq 0\}$ and $n_{j-} = \sum_{i=1}^{n} \ind\{z_{j}< 0\}\ind\{z_{-j}\geq 0\}$, and denote by $\frac{n_{j+}}{n_j} \hat{f}_{j+,p}(z_j)$ and $\frac{n_{j-}}{n_j} \hat{f}_{j-,p}(z_j)$ the estimators of conditional density $f_j(z_j) = f_{Z_j|Z_{-j}}(z_j|z_{-j}\geq0)$ computed considering only observations in $\{z: z \geq 0 \}$ and $\{z: z_j <0, z_{-j} \geq 0 \}$, respectively.

Consider $\theta_j = \lim_{z_j\rightarrow 0^+} f_j(z_j) - \lim_{z_j\rightarrow 0^-} f_j(z_j)$, and note that, when the implication in Equation \eqref{eq:condition_marginal} is true, $\theta_j = 0$ for all $j$. Define the statistic $\hat\theta_{j,p}$:
\begin{gather} \label{eq:thetaj}
\hat\theta_{j,p} = \frac{n_{j+}}{n_j} \hat{f}_{j+,p}(0) - \frac{n_{j-}}{n_j} \hat{f}_{j-,p}(0).
\end{gather}
The following result derives the asymptotic distribution of the statistic.

\begin{prop} \label{thm:snorm}
{\normalfont (Asymptotic distribution of $\hat{\theta}_{j,q}$)}
Under Assumptions \ref{ass:cdf} and \ref{ass:kernel} holding separately for $\{Z:Z\geq0\}$ and $\{Z:Z_{-j}\geq0,Z_j < 0\}$, with $p=2$, $q=p+1$, $n \min\{h_{j-},h_{j+}\} \rightarrow\infty$, and $n \max\{h_{j-}^{1+2q},h_{j+}^{1+2q}\} \rightarrow 0$, when the implication $\theta_{j}=0$ is true:
\begin{gather}
\frac{1}{\sigma_j} \hat \theta_{j,q} \rightarrow^d \mathcal{N}\left(0, 1 \right)
\end{gather}
where
\begin{gather}
\sigma_j^2 = \frac{\pi_{j+}}{h_{j+} \pi_j n} V_{j+}(0) + \frac{\pi_{j-}}{h_{j-} \pi_j n} V_{j-}(0).
\end{gather}
A consistent estimator $\hat \sigma_j^2$ for $\sigma_j^2$ can be obtained by 
\begin{gather}
\hat \sigma_j^2 =\frac{n_{j+}}{h_{j+} n_j^2} \hat V_{j+,q}(0) + \frac{n_{j-}}{h_{j-} n_j^2} \hat V_{j-,q}(0).
\end{gather}
\end{prop}

Proposition \ref{thm:snorm} is valid for any $j$. I am interested in the asymptotic distribution of the vector $\hat{\theta}=(\hat{\theta}_1,\hat{\theta}_2,\dots,\hat{\theta}_d)$, whose distribution under the null hypothesis of continuity of $f_{Z_j|Z_{-j}}(z_j|z_{-j}\geq0)$ is derived in the next theorem.

\begin{theorem} \label{thm:theta}
{\normalfont (Asymptotic distribution of $\hat{\theta}$)}
Under Assumptions \ref{ass:cdf} and \ref{ass:kernel} holding separately for $\{Z:Z\geq0\}$ and $\{Z:Z_{-j}\geq0,Z_j < 0\}$ for all $j$, with $p=2$, $q=p+1$, $n \min\{h_{j-},h_{j+}\} \rightarrow\infty$ and $n \max\{h_{j-}^{1+2q},h_{j+}^{1+2q}\} \rightarrow 0$ for all $j$, when $\theta=0$,
\begin{gather*}
\hat{\Sigma}^{-\frac{1}{2}}\hat \theta \rightarrow^d \mathcal{N}(0,I).
\end{gather*}
where $\hat{\Sigma}_{jj} = \hat{\sigma}^2_j$ as defined in Proposition \ref{thm:snorm}, and $\hat{\Sigma}_{ji}=0$ for all $i \neq j$.
\end{theorem}

Theorem \ref{thm:theta} shows that, even if the number of observations simultaneously considered by any pair of estimators $\theta_j$ and $\theta_i$ goes to infinite, they are asymptotically independent.

The asymptotic distribution of the quadratic form test statistic $\hat{t}$ defined in Equation \eqref{eq:quadratic_form} is immediately derived from Theorem \ref{thm:theta}. The quadratic form is, in fact, a continuous function of $\hat{\Sigma}^{-\frac{1}{2}}\hat \theta$, and is hence distributed as the sum of $d$ squared independent normals: a $\chi^2$ distribution with $d$ degrees of freedom.

Theorem \ref{thm:theta} allows to derive the asymptotic distribution of any continuous function of the vector $\hat{\Sigma}^{-\frac{1}{2}}\hat \theta$. To test the null hypothesis in Equation \eqref{eq:condition_marginal}, it is useful to consider the $\ell^p$-norm statistic $\lVert \hat{\Sigma}^{-\frac{1}{2}}\hat \theta \rVert_p$: large values of this statistic gives evidence against the null hypothesis, suggesting to choose the critical value $c(\alpha)$ as the $1-\alpha$ quantile of its asymptotic distribution. The continuity of the $\ell^p$-norm ensures that $\lVert \hat{\Sigma}^{-\frac{1}{2}} \hat{\theta} \rVert_p \rightarrow^d \lVert X \rVert_p$, where $X \in \mathbb{R}^d$ is a random vector with distribution $\mathcal{N}(0, I)$. In general, the quantiles of $\lVert X \rVert_p$ and hence the critical value for the test can be obtained through simulations, drawing random samples from $X$ and computing $\lVert X \rVert_p$. The quadratic form can be seen as a special case, where the Euclidean distance ($\ell^2$-norm) is squared to obtain a $\chi^2$ distribution whose quantiles can be analytically computed.

With these test statistic and critical value, the MRDD manipulation test $\phi(\hat t,\alpha)$ defined in Equation \eqref{eq:manipulation_test_phi} is asymptotically valid and consistent: when the null hypothesis in Equation \eqref{eq:condition_marginal} is true, the test has an asymptotic rejection probability of $\alpha$. When the null hypothesis is false, the test asymptotically rejects with probability one. The following corollary formalizes this result.

\begin{cor} \label{thm:coroll}
{\normalfont (Manipulation test)}
Let $H_0$ be the null hypothesis in Equation \eqref{eq:condition_marginal}, and consider the manipulation test $\phi(\hat t,\alpha)$ defined in Equation \eqref{eq:manipulation_test_phi} with the test statistic $\hat{t}$ defined in Equation \eqref{eq:quadratic_form} and as critical value $c(\alpha)$ the $1-\alpha$ quantile of the $\chi^2$ distribution. Under the assumptions of Theorem \ref{thm:theta}, when $H_0$ is true:
\begin{gather}
\lim_{n\rightarrow \infty} P(\phi(\hat t,\alpha)=1) = \alpha.
\end{gather}
When $H_0$ is false:
\begin{gather}
\lim_{n\rightarrow \infty} P(\phi(\hat t,\alpha)=1) = 1.
\end{gather}
\end{cor}

In the next section, the manipulation test is compared to some alternative used in the literature, studying its finite sample properties in terms of power and size control.

\begin{remark}
{\normalfont (Choice of the Test Statistic)}
The manipulation test is asymptotically valid and consistent for any $\ell^p$-norm test statistic, with $p > 0$. Different $\ell^p$-norms imply differences in power against different alternatives. The quadratic form test statistic with $p = 2$ performs well in detecting discontinuities diffused across all running variables. These behaviors are of primary interest: when the treatment's benefits lead agents to manipulate their running variables for eligibility, manipulation is likely widespread. In contrast, if only one running variable is manipulated, other test statistics may be more appropriate. For this case, I discuss the $\ell^\infty$-norm statistic, equivalent to the max-statistic $\hat{t}_m = \max\left( \left| \frac{\hat{\theta}_1}{\hat{\sigma}_1} \right|, \dots, \left| \frac{\hat{\theta}_d}{\hat{\sigma}_d} \right| \right)$, in Appendix \ref{app:maxstat}.
\end{remark}





\section{Alternative approaches} \label{sec:comparison}
It is common for applied papers utilizing the regression discontinuity design to include manipulation tests for their running variables, as highlighted by the survey by \cite{canay2018approximate}. Although there is a lack of a theoretical foundation for the test in the multidimensional case, Table \ref{tab:survey} attests the prevalence of manipulation tests in papers employing the MRDD. These papers typically employ two different approaches: computing multiple tests, one for each running variable separately (Separate Tests, ST); or aggregating the running variables considering the distance of each observation from the boundary, and then running the manipulation test for the distance as the single running variable (Distance as running variable Test, DT). The ST approach does not control the size for the null hypothesis in Equation \eqref{eq:condition_marginal}, while the DT is not consistent against certain alternatives, and is not robust to change in the units of measure. In the following sections, I compare these approaches with the proposed manipulation test (MT). I highlight their limitations and show how they can be adapted to properly test the null hypothesis in Equation \eqref{eq:condition_marginal}.

\begin{table}[]
\centering
\caption{\small Published papers using MRDD. Most studies utilize either separate tests (ST), where each running variable's density continuity is tested individually, or distance as running variable tests (DT), which consider the distance of observations from the boundary as the unique running variable.}
\label{tab:survey}
\begin{tabular}{lccc}
\hline
\multicolumn{1}{c}{Authors (Year)} & \multicolumn{1}{c}{Manipulation Test} &  \multicolumn{1}{c}{ST} & \multicolumn{1}{c}{DT} \\ \hline
 \cite{frey2019cash} & $\times$ &  &  \\
 \cite{matsudaira2008mandatory} & \checkmark & \checkmark &  \\
 \cite{hinnerich2014democracy} & \checkmark & \checkmark &  \\
 \cite{elacqua2016short} & \checkmark & \checkmark &  \\
 \cite{egger2015impact} & \checkmark & \checkmark &  \\
 \cite{evans2017smart} & \checkmark & \checkmark &  \\
 \cite{smith2017giving} & \checkmark & \checkmark &  \\
 \cite{londono2020upstream} & \checkmark & \checkmark &  \\
 \cite{clark2014signaling} & \checkmark &  & \checkmark \\
 \cite{cohodes2014merit} & \checkmark &  & \checkmark \\
 \cite{becht2016does} & \checkmark &  & \checkmark \\ \hline
\end{tabular}
\end{table}

\subsection{Separate Tests (ST)}
The separate tests procedure in the context of MRDD treats each running variable separately and applies existing manipulation tests designed for single-dimensional RDD \citep{mccrary2008manipulation,cattaneo2020simple,bugni2021testing}. In some cases \citep{egger2015impact, evans2017smart, londono2020upstream}, the test is conducted on the conditional marginal densities by considering only the units that meet the threshold for the other running variables. Either way, without accounting for multiple hypotheses testing, the ST is invalid, as it does not control the size for the null in Equation \eqref{eq:condition_marginal}.

\subsection{Multiple Hypotheses Test with Bonferroni Correction (BCT)} \label{sec:bct}
A straightforward fix for the ST is to account for multiple hypotheses testing using Bonferroni correction (BCT). In case the test by \cite{cattaneo2020simple} is employed, the resulting procedure partly overlaps with the test I proposed. To test the implication in Equation \eqref{eq:condition_marginal} at level $\alpha$, statistics $\hat{\theta}_j$ defined in Equation \eqref{eq:thetaj} are used to conduct separate tests for each running variable, with the critical values adjusted for multiple testing (for a review on multiple hypotheses testing, see chapter 9 in \cite{lehmann2022testing}). The null hypothesis of running variable $j$ continuity is tested at significance level $\frac{\alpha}{d}$, where $d$ is the number of running variables. The implication in Equation \eqref{eq:condition_marginal} is rejected if the continuity of any of the running variables is rejected. The correction for the number of hypotheses ensures correct coverage, meaning that the asymptotic family-wise error rate, which is the probability of rejecting one or more true null hypotheses (and hence the probability of rejecting the implication in Equation \eqref{eq:condition_marginal} when it is true), is not greater than $\alpha$. In this context, alternative multiple hypotheses corrections (e.g., stepwise methods or Holm correction) would coincide with the BCT, as rejecting continuity for the density of just one running variable is equivalent to rejecting the implication in Equation \eqref{eq:condition_marginal}.

MT and BCT rely on the same vector of statistics $\hat{\theta}$. Local power analysis can be used to compare the power of the two tests against local alternative hypotheses, letting the discontinuity of the density at the threshold get smaller as the sample size increases. I consider a framework where all the running variables are discontinuous, and the discontinuity is equal to $k/\sqrt{nh}$, such that, asymptotically, $\hat{\theta}_j \rightarrow^d \mathcal{N}(k,1)$ for all $j$. This framework mimics a setting where all the running variables are manipulated to get the treatment: if treatment is desirable (or undesirable), I expect all the agents close to the treatment region to manipulate their running variables to get in (or out) the region.

\begin{figure}
\begin{center}
\resizebox{\textwidth}{!}{
\includegraphics{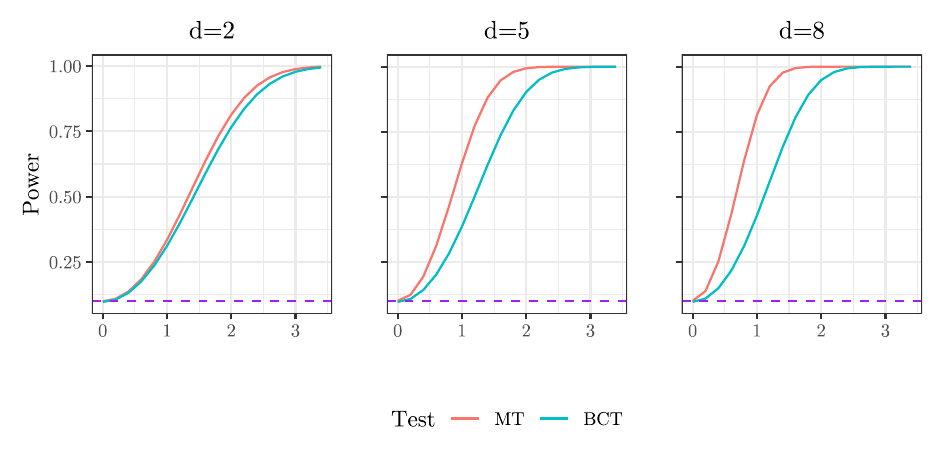}
}
\caption{\small Local asymptotic power curves for the manipulation test proposed in this paper (MT, in red) and the multiple hypotheses test with Bonferroni correction (BCT, in blue). Plots consider different different numbers $d$ of running variables. The significance level $\alpha=0.1$ is represented by the dotted horizontal purple line. All the running variables are discontinuous, and the discontinuity is equal to $k/\sqrt{nh}$, such that $\hat{\theta}_j \rightarrow^d \mathcal{N}(k,1)$ for all $j$.}
\label{fig:local_asymp}
\end{center}
\end{figure}

Figure \ref{fig:local_asymp} reports power curves for MT (in red) and BCT (in blue), considering different numbers of running variables $d \in \{ 2,5,8\}$. When all running variables exhibit discontinuity, MT outperforms BCT in terms of power. This is because MT combines information from all the running variables and effectively detects manipulation when it is widespread across them. On the other hand, BCT considers each running variable separately, which results in lower power in case of widespread manipulation.

The local power analysis confirms that aggregating information and testing a single hypothesis is better than testing multiple hypotheses for the continuity of each running variable separately. The MT is less conservative, at least against alternative hypotheses where manipulation is spread across all running variables.


\subsection{Distance as Running Variable Test (DT)} \label{ssec:srt}
The second approach for manipulation tests in the multidimensional setting employed in applied research involves dimension reduction: the multidimensional regression discontinuity design is reduced to a single-dimensional design, with the scalar distance between the vector of running variables and the boundary $\mathcal{B}$ as the only running variable. The distance is used to estimate the conditional average treatment effect, similar to the classical RDD, and to conduct a manipulation test using one of the available tests \citep{mccrary2008manipulation,cattaneo2020simple,bugni2021testing}. This approach appears simple since it directly relates to the single-dimensional RDD case. Nonetheless, it comes with some caveats that need to be considered.

First, the choice of distance metric and measurement units can significantly impact the test results. Different distance metrics used to measure the distance between the running variables and the boundary lead to different test statistics and test outcomes, as well as different units of measurement for the running variables. To address this second issue, one possible solution is to standardize the running variables to have unit variance before conducting the manipulation test, but this practice may in fact worsen the properties of the test (as showed in Monte Carlo simulations).

A second flaw of the DT is that it is inconsistent against certain fixed alternatives of the null hypothesis in Equation \eqref{eq:test}. For instance, if there are opposite discontinuities in the marginal distributions of different running variables, and these discontinuities balance each other out, the asymptotic probability that the DT will reject the false null hypothesis in Equation \eqref{eq:test} is equal to $\alpha$, rather than one. A design where the DT test is inconsistent is studied in Section \ref{sec:mod3}.

Despite the lack of a clear theoretical background and the lack of power against specific alternative hypotheses, it may still be the case that the DT performs better than MT and BCT in contexts plausible in applications. However, if any, the evidence suggests poorer finite sample performance for DT, as shown by \cite{wong2013analyzing} or in the Monte Carlo simulation in the following section.

\section{Monte Carlo simulation} \label{sec:mc}
In this section, I conduct Monte Carlo simulations to investigate the finite sample performance of the manipulation test (MT) proposed in this paper. It is compared with the multiple hypotheses test with Bonferroni correction (BCT) and the single running variable test, both with standardization to unitary variance (SDT) and without it (DT).

Without loss of generality, the cutoff is set at 0 for all running variables: units are treated when all the running variables are nonnegative. The boundary is the set of points with all nonnegative coordinates and at least one coordinate equal to zero. Figure \ref{fig:designs} reports a realization of the simulated samples for the four models, illustrating the joint distribution of the running variables.

Models \ref{mod:unif} and \ref{mod:norm} show how the tests are comparable in controlling the size, while Models \ref{mod:divdisc} and \ref{mod:disc} attest the better power properties for MT discussed in Section \ref{sec:comparison}.

\begin{figure}
\begin{center}
\resizebox{\textwidth}{!}{
\includegraphics{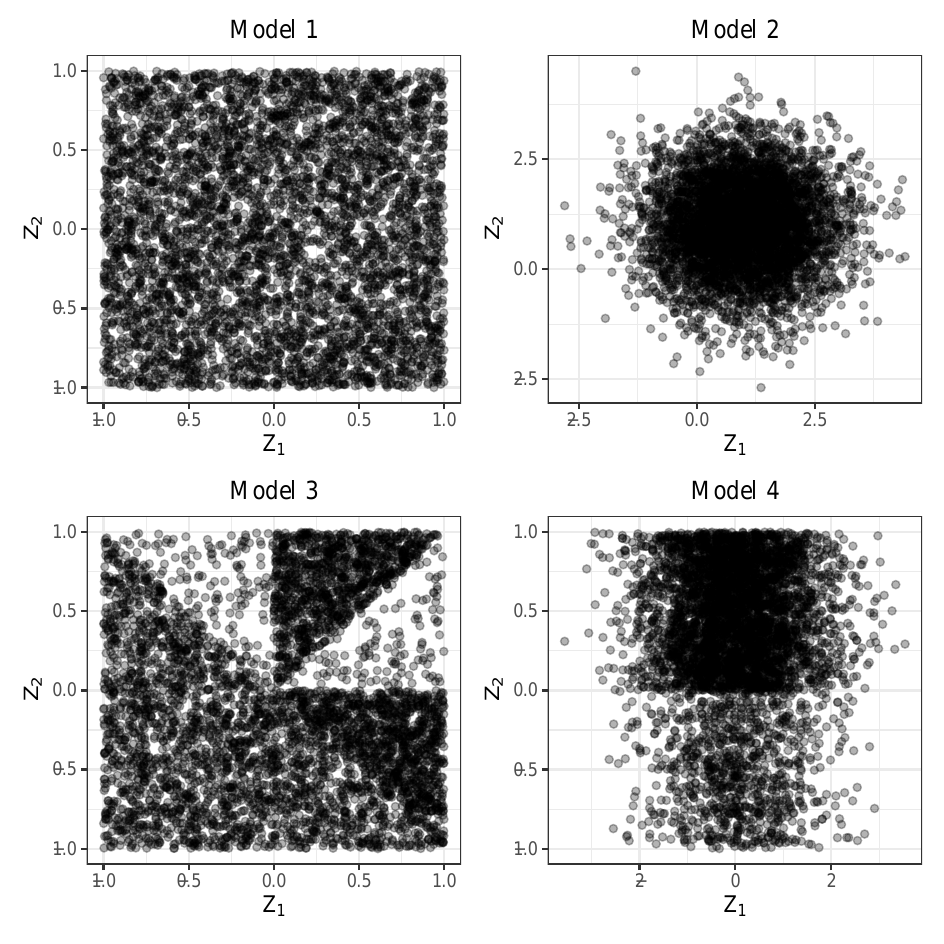}
}
\caption{\small Scatterplots of a sample of size $n=5,000$ from the four models illustrate the joint distribution of $Z_1$ and $Z_2$. For Models \ref{mod:unif} and \ref{mod:norm}, parameter $d$ is set to 2 (two running variables). Joint density is continuous, and the condition in Equation \eqref{eq:condition_marginal} is satisfied. For Models \ref{mod:divdisc} and \ref{mod:disc}, parameters $\gamma_1$ and $\gamma_2$ are 0.8: joint density is not continuous, and condition in Equation \eqref{eq:condition_marginal} is not satisfied.}
\label{fig:designs}
\end{center}
\end{figure}

\subsection{Model 1 and Model 2}
\begin{modl} \label{mod:unif}
Consider $d$ running variables uniformly distributed:
\begin{gather}
    Z_j \sim U(-1,1) \quad \text{for } j \text{ in } \{1,\dots,d \}.
\end{gather}
\end{modl}
In Model \ref{mod:unif}, densities are symmetrical to the threshold, and the density function is flat. Since the setting can be particularly convenient for the tests, Model \ref{mod:norm} considers densities with different behaviors at the two sides of the boundary.

\begin{modl} \label{mod:norm}
Consider $d$ running variables normally distributed and centered at 1:
\begin{gather}
    Z_j \sim \mathcal{N}(1,1) \quad \text{for } j \text{ in } \{1,\dots,d \}.
\end{gather}
\end{modl}

Both Models \ref{mod:unif} and \ref{mod:norm} are simulated considering sample sizes $n \in \{500,2000,5000 \}$ and total numbers of running variables $d \in \{2,3,4 \}$. The simulation results are presented in Table \ref{tab:mod12}. Overall, for all the sample sizes considered, the tests tend to under-reject, with empirical rejection rates closer to the theoretical ones for DT and SDT. MT and BCT exhibit similar performances across different models and parameters specification: rejection rates get closer to asymptotic ones as the effective sample size grows, when $d$ decreases for a fixed $n$ or $n$ increases for a fixed $d$. For the same values of parameters $d$ and $n$, under-rejection is larger in Model \ref{mod:unif} than Model \ref{mod:norm}: as expected, the steeper is the probability density function at the cutoff, the higher is the probability for the test to reject the true null.

\begin{table}[]
\centering
\caption{\small Rejection rates under the true null hypothesis of continuity of marginal densities of the running variables, computed through 5,000 Monte Carlo simulations, at 5\% significance level. MT is the manipulation test proposed in this paper; BCT is the multiple hypotheses test with Bonferroni correction; DT and SDT consider as single running variables the Euclidean distance from the boundary: for SDT, running variables are standardized to have unitary variance before computing the distance, while for DT they are not. $d$ is the number of running variables, and $n$ is the sample size.}
\label{tab:mod12}
\begin{tabular}{ll @{\hspace{3em}} cccc @{\hspace{2.5em}} cccc}
\hline \hline
 &  & \multicolumn{4}{c}{Model 1} & \multicolumn{4}{c}{Model 2} \\ \hline
$d$ & $n$ & MT & BCT & DT & SDT & MT & BCT & DT & SDT \\ 
\hline 
\multirow{3}{*}{2} & 500 & $0.025$ & $0.030$ & $0.043$ & $0.044$ & $0.037$ & $0.036$ & $0.043$ & $0.043$ \\ 
& 2,000 & $0.036$ & $0.039$ & $0.047$ & $0.046$ & $0.042$ & $0.041$ & $0.050$ & $0.049$ \\ 
& 5,000 & $0.032$ & $0.031$ & $0.045$ & $0.045$ & $0.040$ & $0.039$ & $0.052$ & $0.052$ \\
\hline
\multirow{3}{*}{3} & 500 & $0.024$ & $0.025$ & $0.036$ & $0.037$ & $0.029$ & $0.033$ & $0.041$ & $0.040$ \\ 
& 2,000 & $0.027$ & $0.029$ & $0.041$ & $0.040$ & $0.038$ & $0.039$ & $0.039$ & $0.039$ \\ 
& 5,000 & $0.033$ & $0.033$ & $0.037$ & $0.037$ & $0.037$ & $0.040$ & $0.045$ & $0.044$ \\ 
\hline
\multirow{3}{*}{4} & 500 & $0.025$ & $0.019$ & $0.047$ & $0.047$ & $0.037$ & $0.032$ & $0.043$ & $0.039$ \\ 
& 2,000 & $0.020$ & $0.019$ & $0.041$ & $0.040$ & $0.036$ & $0.036$ & $0.042$ & $0.042$ \\ 
& 5,000 & $0.030$ & $0.035$ & $0.044$ & $0.044$ & $0.039$ & $0.041$ & $0.041$ & $0.041$ \\ 
 \hline \hline 
\end{tabular}
\end{table}

Unsurprisingly, all the tests have a comparable performance: no theoretical reason suggests discrepancies for the three tests in controlling size. Differences arise when finite sample power is studied, as shown by Models \ref{mod:divdisc} and \ref{mod:disc}.

\subsection{Model 3} \label{sec:mod3}

\begin{modl} \label{mod:divdisc}
Define random vector $Z^* =(Z_1^*,Z_2^*)$, where $Z_1^* \sim U(-1,1)$, $Z_2^* \sim U(-1,1)$, and $Z_1^*$ and $Z_2^*$ independent. Define sets $A_1= \{(z_1,z_2): z_1<0, -z_1<z_2 \}$ and $A_2= \{(z_1,z_2): z_1>z_2, z_2>0 \}$.

Consider two running variables $Z_1$ and $Z_2$ distributed as follows:
\begin{gather}
    Z_1 \sim \begin{cases}
Z_1^*, &  \text{if } Z^* \not\in A_1\\
Z_1^*, & \text{with probability $1-\gamma_1$} \text{ if } Z^* \in A_1\\
-Z_1^*, & \text{with probability $\gamma_1$} \text{ if } Z^* \in A_1
\end{cases} \\
Z_2 \sim \begin{cases}
Z_2^*, &  \text{if } Z^* \not\in A_2\\
Z_2^*, & \text{with probability $1-\gamma_1$} \text{ if } Z^* \in A_2\\
-Z_2^*, & \text{with probability $\gamma_1$} \text{ if } Z^* \in A_2
\end{cases}
\end{gather}
\end{modl}

Model \ref{mod:divdisc} mimics a setting where the two running variables are manipulated, but in opposite directions: when $Z^* \in A_1$, $Z_1$ is manipulated to get the treatment; when $Z^* \in A_2$, $Z_2$ is manipulated to avoid the treatment.
Parameter $\gamma_1$ governs the extent of manipulation: when $\gamma_1=0$, the joint density of $Z_1$ and $Z_2$ is continuous; when $\gamma_1=1$, the joint density becomes zero in regions $A_1$ and $A_2$, resulting in the maximum discontinuity.

The curves depicted in Figure \ref{fig:pcurves} illustrate the finite sample performance of the tests. For both MT and BCT, the power of the tests increases with the degree of manipulation $\gamma_1$, as expected. For DT and SDT, the power always remains equal to the test size. This design corresponds to the situation described in Section \ref{ssec:srt}: the condition in Equation \eqref{eq:condition_marginal} is not satisfied, since neither the marginal densities of $Z_1$ nor $Z_2$ are continuous at the threshold (as shown in Figure \ref{fig:designs}). Nonetheless, the probability density function of the distance from the boundary is continuous. Consequently, the null hypothesis tested by DT and SDT is true, resulting in trivial power for these tests.

\begin{figure}
\begin{center}
\resizebox{\textwidth}{!}{
\includegraphics{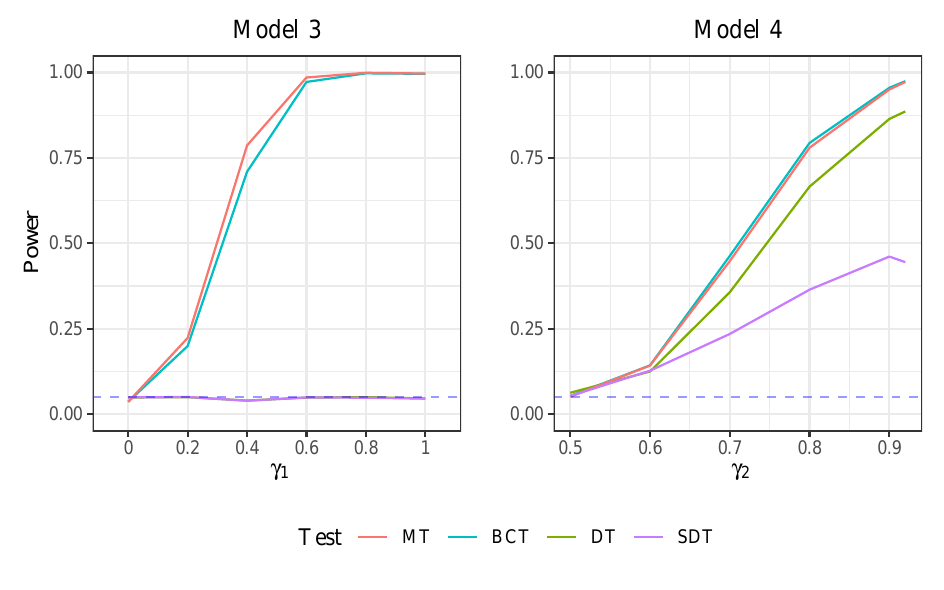}
}
\caption{\small Power of different manipulation tests with $n=2,000$, computed though 5,000 Monte Carlo simulations. The dotted line indicates the nominal size of the tests (5\%). MT is the manipulation test proposed in this paper; BCT is the multiple hypotheses test with Bonferroni correction; DT and SDT consider as single running variables the Euclidean distance from the boundary: for SDT, running variables are standardized to have unitary variance before computing the distance; for DT, they are not. Parameters $\gamma_1$ and $\gamma_2$ determine the degree of manipulation. In Model 3, lines for DT and SDT overlap.}
\label{fig:pcurves}
\end{center}
\end{figure}

\subsection{Model 4}

\begin{modl} \label{mod:disc}
Consider two running variables distributed as follows:
\begin{gather}
    Z_1 \sim \mathcal{N}(0,1) \\
    Z_2 \sim \begin{cases}
U(0,1), & \text{with probability $\gamma_2$}\\
U(-1,0), & \text{with probability $1-\gamma_2$}.
\end{cases}
\end{gather}
\end{modl}
Model \ref{mod:disc} is a design where only $Z_2$ is manipulated. The manipulation is determined by the parameter $\gamma_2$. When $\gamma_2=0.5$, $Z_2$ has a continuous density, following a uniform distribution between -1 and 1. When $\gamma_2=1$, the density of $Z_2$ becomes zero at the left of the boundary. The degree of discontinuity increases as the value of $Z_2$ deviates further from 0.5.

The curves in Figure \ref{fig:pcurves} show how the finite sample power depends on $\gamma_2$. The power increases with higher values of $\gamma_2$ for all tests, but it is lower for DT and SDT. Additionally, despite two similar versions of the same test, DT and SDT's performances are different. As discussed in Section \ref{ssec:srt}, the choice of the unit of measure affects the result of DT.
In this case, the standardization applied in SDT reduces its power compared to DT. Standardization is not the solution to the issue.

MT and BCT exhibit similar behavior in the context of Model \ref{mod:disc}. It mimics a framework different from the one studied in the local asymptotic analysis, and there are no theoretical reason to expect the MT to perform better in this setting.

Overall, the Monte Carlo simulations confirm that the manipulation test proposed in this paper has better finite sample properties than alternative tests. The simulations demonstrate advantages in terms of power and robustness, reinforcing the findings derived from the local asymptotic analysis discussed earlier.

As outlined in Section \ref{sec:test}, the proposed manipulation test can be readily implemented using existing packages in popular statistical software, such as R and Stata, with just a few lines of code. This ease of implementation enhances the practical applicability of the test. The next section implements the manipulation test in a real-world application, illustrating its simplicity.

\section{Application: \cite{frey2019cash}} \label{sec:app}
I apply my manipulation test for the MRDD considered by \cite{frey2019cash} investigating the political economy of redistributive policies. In the original analysis, no manipulation test is reported. The paper studies the impact of cash transfers implemented by the Brazilian federal government on the dynamics of clientelism at the municipal level. The main hypothesis suggests that these cash transfers, by reducing the vulnerability of the poor, diminish the attractiveness of clientelism as a strategy for incumbent mayors.

The Bolsa Família (BF) program is the largest conditional cash transfer program globally and has been implemented in households across Brazil since 2003. The coverage of BF across different municipalities exhibits a positive correlation with the funding allocated to the Family Health Program (FHP), a household-based healthcare program run by municipalities since 1995. The positive correlation between BF coverage and FHP funding can be attributed to the fact that FHP teams have a significant penetration among the poor households, potential beneficiaries of BF. This enables them to effectively disseminate information about the BF program and encourage enrollment among eligible households.

To estimate the causal effect of the cash transfers on local clientelism, \cite{frey2019cash} exploits the link between BF and FHP, along with a specific discontinuity in the design of the FHP. The FHP provides municipalities with an additional 50\% funding if they meet two criteria: a population of fewer than 30,000 inhabitants and a Human Development Index (HDI) below 0.70. This discontinuity, determined by the joint thresholds of population and HDI, is directly reflected in the diffusion of the BF program: consequences of cash transfers can be analyzed using an MRDD.

Figure \ref{fig:frey2019} provides a visualization of the MRDD. In the space of the two running variables (population and HDI), treated municipalities are depicted in light blue, and untreated municipalities in dark blue. The red line represents the boundary. In this specific context, the treatment corresponds to the additional FHP funding, which leads to variations in the adoption rate of the BF program.

\begin{figure}
\begin{center}
\resizebox{\textwidth}{!}{
\includegraphics{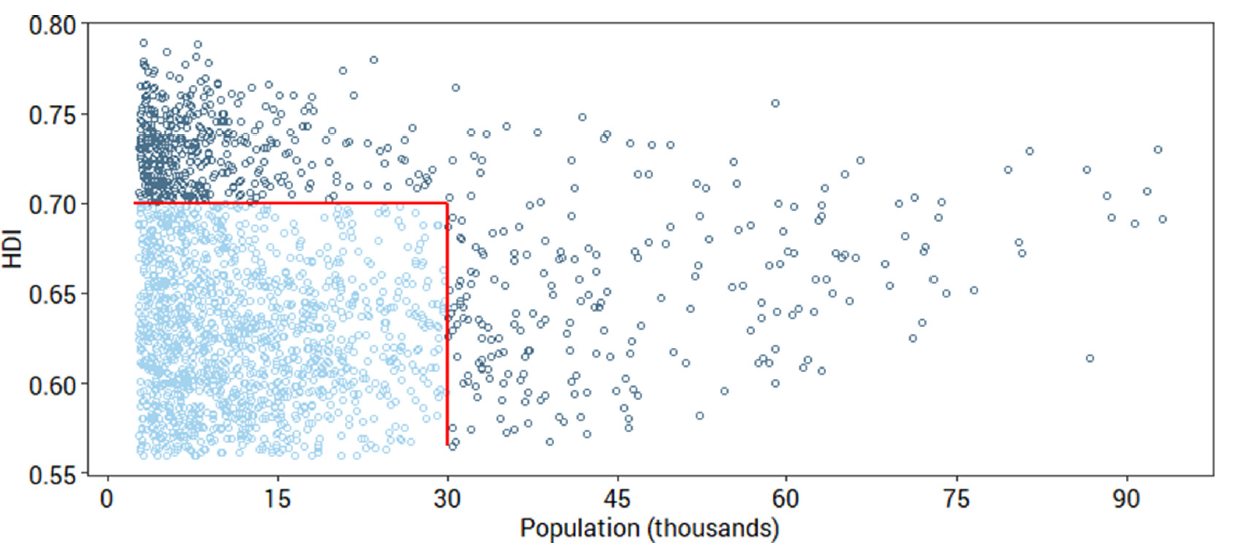}
}
\caption{\small Running variables for the MRDD considered by \cite{frey2019cash}. Municipalities are assigned to the policy (light blue dots) when the population is below 30,000 inhabitants (x-axis), and Human Development Index is below 0.7 (y-axis).}
\label{fig:frey2019}
\end{center}
\end{figure}

In this context, the MRDD requires the assumption that the joint density of population and HDI is continuous at the boundary, thereby ensuring the continuity of their respective marginal densities. The manipulation test is employed to validate the design and enhance the credibility of the study's findings. The test is conducted considering two running variables, resulting in a p-value of 0.490, as reported in Table \ref{tab:pvalues}. With a significance level of $\alpha = 0.05$, the null hypothesis is not rejected, indicating no evidence of manipulation. In this context, the same conclusion of absence of manipulation is also reached by the other tests discussed in Section \ref{sec:comparison} (BCT, DT, and SDT). This is not surprising, as all the tests are expected to control the size similarly and mostly differ in power. It is nonetheless interesting to observe how, even in real applications, the DT and SDT yield different p-values, despite using the same procedure just with a rescale of the running variables.

\begin{table}[]
\centering
\caption{\small P-values for various manipulation tests for the MRDD considered by \cite{frey2019cash}.}
\label{tab:pvalues}
\begin{tabular}{l @{\hspace{3em}} cccc }
\hline \hline
 Test & MT & BCT & DT & SDT \\ 
\hline 
 p-value & $0.49$ & $0.35$, $0.45$ & $0.19$ & $0.38$  \\ 
 \hline \hline 
\end{tabular}
\end{table}

It is important to emphasize that the manipulation test alone does not establish the model's validity. It serves as a robustness check and provides supporting evidence for the continuity of the densities of the running variables, which is a necessary condition for Assumption \ref{ass:ycont}. It cannot substitute a discussion on why the assumptions of the MRDD are likely to hold in this setting, a discussion that remains essential for drawing valid conclusions from the analysis.

\section{Conclusion} \label{sec:fut}
This paper introduced a manipulation test for the multidimensional regression discontinuity design. Like the single-dimensional RDD developed by \cite{lee2008randomized}, I constructed a causal inference model for the MRDD. From an assumption on unobservable agents' behavior, I derived an implication on observable quantities—the continuity of the conditional marginal densities of the multiple running variables. I proposed a manipulation test for the implication, which was then compared with alternative approaches commonly used in applied research. While these approaches vary, they generally lack clear theoretical justification, and some are inconsistent for the considered implication. Through Monte Carlo simulations and local power analysis, I explored the finite sample properties of the tests.

The manipulation test should be seen as a robustness check to strengthen the credibility of the assumptions required by the MRDD, and it is not intended as a pre-test. It can be readily implemented with already existing packages without the need for additional tuning parameters. Considering the application in \cite{frey2019cash}, I showed how to use the test in practice.

\bibliographystyle{chicago}
\bibliography{references}

\newpage

\appendix
\counterwithin{figure}{section}

\section{Model for the MRDD}\label{app:model}

In this Appendix, I extend the model proposed by \cite{lee2008randomized} for the Regression Discontinuity Design to the multidimensional setting, allowing the number of running variables $d$ to be larger than one. Some of the results are not new (identification results already appeared in \cite{imbens2009regression} and \cite{keele2015geographic}), but it is useful to report the entire extension to show where the implication tested by the manipulation test comes from.

Let $W$ be a random variable with support $\mathcal{W}$ and density $g(w)$. $W$ indicates the type of agents, unobservable to the researcher, and can be discrete or continuous, with finite or infinite support.

Let $Z\in \R^d$ be a random vector of $d$ observable continuous running variables with joint cumulative distribution function $F(z)$ and marginal CDFs $F_j(z_j)$.

The treatment status $D$ depends on $Z$: I consider a sharp design where the treatment status depends deterministically on $Z$. The model can also be extended to a fuzzy design, where the probability of being treated changes discontinuously at the threshold. Units are treated ($D=1$) when all components of $Z$ are above a certain threshold:
\begin{gather}
D = D(Z) = \ind \{Z \geq c \} = \ind \{Z_1 \geq c_1 \}\ind \{Z_2 \geq c_2 \}\dots\ind \{Z_d \geq c_d \}
\end{gather}
and $D=0$ otherwise. Without loss of generality, consider a rescaling of $Z$ such that $c_j=0$, for all $j$: units are treated when, for all $j$, $Z_j\geq0$.

Let $\mathcal{T}$ be the set of values of $Z$ for which $D(Z)=1$, and indicate with $\bar{\mathcal{T}^C}$ the closure of the complement of $\mathcal{T}$. Define the boundary $\mathcal{B} = \mathcal{T} \cap \bar{ \mathcal{T}^C}$, and note that, by definition, in a neighborhood of any point $b\in \mathcal{B}$ there are both treated and untreated units.

Let $F(z|w):\R^d \times \mathcal{W} \rightarrow [0,1]$ be the cdf of $Z$ conditional on $W$.

\begin{ass}\label{ass:cont}
{\normalfont (Continuity of density)}
For all $w \in \mathcal{W}$, for all $b \in \mathcal{B}$, $F(z|w)$ is continuously differentiable in $z$ at $b$. Furthermore, the conditional density $f(z|w)$ is bounded away from zero at $b$.
\end{ass}
Assumption \ref{ass:cont} is asking the conditional density $f(z|w)$ of the running variables $Z$ to be continuous and non-zero at $b\in \mathcal{B}$, for all $w$. The testable implication for the manipulation test is derived from this assumption. It does not entirely exclude influence over the score vector $Z$: some manipulation is allowed, as long as it is not deterministic, as discussed in Section \ref{sec:example}.

The researcher observes an iid sample from the joint distribution of $\{Z,D,Y\}$, where the observed outcome $Y$ is defined as $Y=(1-D)Y_0 + DY_1$. Potential outcomes $Y_0 = y_0(W,Z)$ and $Y_1 = y_1(W,Z)$ are random functions of $W$ and $Z$. For each value of agent type $W$ and running variable $Z$, $y_0(w,z)$ and $y_1(w,z)$ are random variables, and their distributions satisfy the following assumption.

\begin{ass} \label{ass:ycont}
{\normalfont (Continuity of expectation of potential outcomes)} For all $w \in \mathcal{W}$, for all $b \in \mathcal{B}$, expected potential outcomes $\E[y_0(w,z)]$ and $\E[y_1(w,z)]$ are continuous in $z$ at $b$.
\end{ass}

Intuitively, the MRDD compares treated and untreated units close to the boundary. Assumption \ref{ass:ycont} guarantees that the local information from the observable quantities on one side of the boundary is informative about the unobservable quantities on the other side.

The parameters of interest are $\tau(b)$ and $\tau$. $\tau(b)$ is the Conditional Average Treatment Effect (CATE):
\begin{gather*}
\tau(b) = \E[Y_1 -Y_0 |Z=b], \quad b \in \mathcal{B}
\end{gather*}
and $\tau$ is the Integrated Conditional Average Treatment Effect (ICATE):
\begin{gather*}
\tau = \E[Y_1 -Y_0 |Z \in \mathcal{B}].
\end{gather*}
The CATE is a function that maps every point on the boundary to an average treatment effect, while the ICATE aggregates these effects into a single parameter.

Theorem \ref{thm:ident} is the main result for identification in the multidimensional setting.

\begin{theorem} \label{thm:ident}
{\normalfont (Identification of $\tau(b)$ and $\tau$)}
Under Assumptions \ref{ass:cont} and \ref{ass:ycont}, $\tau(b)$ and $\tau$ are identified.
\end{theorem}

The result in Theorem \ref{thm:ident} can be seen as a generalization of \cite{lee2008randomized} that does not require $\mathcal{B}$ to be a singleton and specifies new parameters of interest, suitable for the MRDD (in the single-dimensional case, $\tau(b)$ and $\tau$ coincide).

\subsection{Example} \label{sec:example}

The example depicted in Figure \ref{fig:mrddintro} is helpful to gain some intuition of the model. A scholarship is assigned to students who score above certain thresholds in both math and language tests, and a researcher is interested in the average effect of the scholarship on, for example, the probability of college admission. The scholarship is not randomly assigned, and comparing the average college attendance rates between the treated and untreated groups can be misleading. Even with no effects of the scholarship, for example, we may expect higher college attendance for students who did well on the math and language tests, because of unobservable characteristics, such as effort, correlated with the scores. Consider, however, a combination of math and language scores on the solid purple boundary in the picture. All students are similar in a neighborhood of those scores, except for their treatment status. To compute the CATE $\tau(b)$ of the scholarship for students with that specific combination of math and language scores, the researcher may compare locally treated and untreated units close to that score. The procedure can be applied to every combination of math and language scores at the boundary, and the estimated effects, dependent on the considered point of the boundary, can be aggregated in a unique average treatment effect, the ICATE $\tau$.

I mentioned that Assumption \ref{ass:cont} does not entirely exclude influence over the score vector $Z$: some manipulation is allowed, as long as it is not deterministic. For example, Assumption \ref{ass:cont} is satisfied if, whatever the type $w$, $Z$ can be decomposed as $Z=V + \epsilon$, with $V$ and $\epsilon \in \R^d$, $V$ a deterministic part, perfectly controllable by agents, and dependent on the type $w$, and $\epsilon$ a random component with continuous density. In the scholarship context, imagine two student types: $w_1$ and $w_2$. Students type $w_1$ are not interested in the scholarship, which does not influence their behavior. Assumption \ref{ass:cont} is satisfied for them. Students type $w_2$, instead, want to secure the scholarship but also know that their score in language or math is below the threshold. They decide to study harder to improve their score and win the scholarship. Since studying is costly and they want to minimize their effort, they choose to improve their score to a level $V$ just at the boundary. The density of $V$ is not continuous at $\mathcal{B}$. However, the test score may have a random component $\epsilon$ that students cannot perfectly control: it may be their focus on the test day. Despite all students type $w_2$ having the same $V$, not all get the same score $Z$, and as long as $\epsilon$ has continuous density, $f(z|w_2)$ is continuous at the boundary $\mathcal{B}$. Even in the case of stochastic manipulation, hence, Assumption \ref{ass:cont} is satisfied.

\section{Formulas for $B(x)$, $V(x)$, and $\hat{V}(x)$} \label{app:formulas}
Formulas for $B(x)$, $V(x)$, and $\hat{V}(x)$ are derived by \cite{cattaneo2020simple}.

Let $x_L$ and $x_U$ indicate the lower and the upper bound of the support of $X$: the support does not need to be bounded, and they can be $-\infty$ and $\infty$. First, define the following:
\begin{align*}
    A(x) =& f(x) \int_{\frac{x_L-x}{h}}^{\frac{x_U-x}{h}} r_p(u)r_p(u)'K(u) du \\
    a(x) =& f(x) \frac{F^{(p+1)}(x)}{(p+1)!} \int_{\frac{x_L-x}{h}}^{\frac{x_U-x}{h}} u^{p+1} r_p(u)K(u) du \\
    C(x) =& f(x)^3 \int_{\frac{x_L-x}{h}}^{\frac{x_U-x}{h}} \int_{\frac{x_L-x}{h}}^{\frac{x_U-x}{h}} \min\{u,v\} r_p(u)r_p(v)'K(u)K(v) du dv.
\end{align*}
with $A(x) \in \R^{(p+1)\times(p+1)}$, $a(x) \in \R^{(p+1)}$, $C(x) \in \R^{(p+1)\times(p+1)}$. Bias $B(x)$ and variance $V(x)$ are:
\begin{align*}
    B(x) &= e_1' A(x)^{-1} a(x) \\
    V(x) &= e_1' A(x)^{-1} C(x) A(x)^{-1} e_1.
\end{align*}
Consider the following estimators for $A(x)$ and $C(x)$:
\begin{align*}
    \hat{A}(x) =& \frac{1}{nh} \sum_{i=1}^{n} r_p\left(\frac{x_i -x}{h}\right) r_p\left(\frac{x_i -x}{h}\right)'K\left(\frac{x_i -x}{h}\right) \\
    \hat{C}(x) =& \frac{1}{n^3 h^3} \sum_{i,j,k=1}^{n} r_p\left(\frac{x_j -x}{h}\right) r_p\left(\frac{x_k -x}{h}\right)' K\left(\frac{x_j -x}{h}\right) K\left(\frac{x_k -x}{h}\right) \\
    & \left[ \ind\{x_i \leq x_j\} - \hat{F}(x_j) \right] \left[ \ind\{x_i \leq x_k\} - \hat{F}(x_k) \right]
\end{align*}
and then the estimator for $V(x)$:
\begin{align*}
    \hat{V}_p(x) &= e_1' \hat{A}(x)^{-1} \hat{C}(x) \hat{A}(x)^{-1} e_1.
\end{align*}
Consistency of $\hat{V}_p(x)$ for $V(x)$ is proved in Theorem 2 in \cite{cattaneo2020simple}.

\section{Test with the Max-statistic} \label{app:maxstat}
From the result in Theorem \ref{thm:theta} it is possible to derive the asymptotic distribution of any continuous function applied to the statistic $\hat{\Sigma}^{-\frac{1}{2}} \hat{\theta}$, as the $\ell^p$-norm statistic $\lVert \hat{\Sigma}^{-\frac{1}{2}}\hat \theta \rVert_p$ that I discussed in Section \ref{subsec:manipulation_test}. An interesting case is the $\ell^\infty$-norm, which can be use to construct the manipulation test $\phi_m(\hat{t}_m,\alpha)$ defined as follows:
\begin{gather}
\phi_m(\hat{t}_m,\alpha) = \begin{cases}
1, & \text{if $\hat{t}_m > c_m(\alpha)$}\\
0, & \text{if $\hat{t}_m \leq c_m(\alpha)$}.
\end{cases}
\end{gather}
Here, $\hat{t}_m$ is the max-statistic $\hat{t}_m = \lVert \hat{\Sigma}^{-\frac{1}{2}}\hat \theta \rVert_\infty = \max\left( |\frac{\hat{\theta}_1}{\hat{\sigma}_1}|, \dots, |\frac{\hat{\theta}_d}{\hat{\sigma}_d}| \right)$, and the critical value $c_m(\alpha)$ is the $1-\alpha$ quantile of the distribution of $\max\left( |X_1|, \dots, |X_d| \right)$, where $X \sim \mathcal{N}(0,I_d)$.

Both the MT (the manipulation test using the quadratic form $\hat{t}$ as test statistic) and $\phi_m$ are consistent tests for the implication in Equation \eqref{eq:condition_marginal}, but their power against alternative hypotheses differs. The MT considers an average of statistics $\left( \frac{\hat{\theta}_j}{\hat{\sigma}_j} \right)^2$, and is hence better at rejecting the null hypothesis when manipulation is spread across all the running variables. Taking the maximum over $|\frac{\hat{\theta}_j}{\hat{\sigma}_j}|$, instead, $\phi_m$ has greater power when manipulation occurs for only one running variable.

In Section \ref{sec:bct}, I compared the power of the MT and the BCT when all the running variables are manipulated. The same local power analysis can be used to study an alternative framework, where only one running variable ($j=1$) is discontinuous, such that $\hat{\theta}_1 \rightarrow^d \mathcal{N}(k,1)$ and $\hat{\theta}_j \rightarrow^d \mathcal{N}(0,1)$ for $j\neq1$. This alternative framework depicts a situation where manipulation only occurs for one running variable: this may happen because running variables are different, and some could be impossible to manipulate. In this case, we expect the MT to perform poorly, compared to the test with the max-statistic.

The local power analysis confirms the intuition. Figure \ref{fig:pc_app} shows that in this framework the MT exhibits a lower power than the BCT, as combining information it dilutes the signal from the discontinuous running variable. Nonetheless, the manipulation test $\phi_m$ dominates the BCT: properly aggregating information from all the running variables and testing a unique hypothesis is better than considering separate tests for each $\hat{\theta}_j$ accounting for multiple tests.

\begin{figure}
\begin{center}
\resizebox{\textwidth}{!}{
\includegraphics{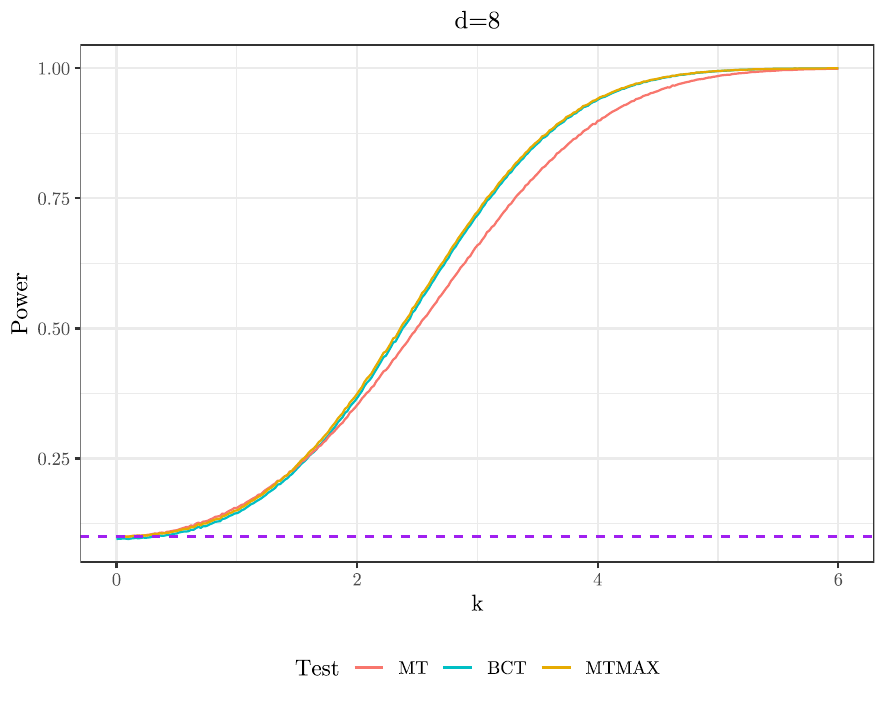}
}
\caption{\small Local asymptotic power curves for the manipulation test proposed in this paper (MT, in red) and the multiple hypotheses test with Bonferroni correction (BCT, in blue), and the manipulation test $\phi_m$ considering the max-statistic (MTMAX, in yellow). The dotted horizontal purple line represents the significance level $\alpha=0.1$. The data distribution is such that only running variable $j=1$ is discontinuous, and $\hat{\theta}_1 \rightarrow^d \mathcal{N}(k,1)$ and $\hat{\theta}_j \rightarrow^d \mathcal{N}(0,1)$ for $j\neq1$.}
\label{fig:pc_app}
\end{center}
\end{figure}

\section{Proofs}
\setcounter{prop}{0}
\setcounter{theorem}{0}
\setcounter{cor}{0}

\subsection*{Proposition \ref{prop:impl}}
\begin{prop}
{\normalfont (Testable implication)}
    Assumption \ref{ass:cont} implies the following condition on observable quantities:
    \begin{gather} \label{eq:test}
    f_{Z_j|Z_{-j}}(z_j|z_{-j}\geq0) \text{ continuous at } z_j=0, \forall j.
    \end{gather}
\end{prop}
\begin{proof}
    Consider the density of $Z$, $f(Z)$, that can be written as $f(z) = \int_{\mathcal{W}} f(z | w) g(w) dw$.

    By Assumption \ref{ass:cont}, $f(z|w)$ is continuous at $z=b\in\mathcal{B}$, and hence $f(z) \text{ is continuous at } z=b, \forall b \in \mathcal{B}$. The implication is not sharp: $f(z)$ may be continuous even if $f(z|w)$ is not (see \cite{lee2008randomized} for a further discussion on this issue, which analogously arises in the single-dimensional case).
    
    For any $j$, let $f_{-j}(z_{-j})$ be the joint density of $z_{-j} = (z_1,\dots,z_{j-1},z_{j+1},\dots,z_d) \in \R^{d-1}$. By definition of conditional density:
    \begin{gather}
    f_{Z_j|Z_{-j}}(z_j|z_{-j}) = \frac{f_{Z_j,Z_{-j}}(z_j,z_{-j})}{f_{-j}(z_{-j})} =
    \frac{f(z)}{f_{-j}(z_{-j})}
    \end{gather}
    and since $f(z)$ is continuous in $z$ and hence in $z_j$ at all $b \in \mathcal{B}$, so it is $f_{Z_j|Z_{-j}}(z_j|z_{-j})$. Whenever $z_{-j} \geq 0$, $f_{Z_j|Z_{-j}}(z_j|z_{-j})$ is continuous at $z_j=0$, as the set $\{z:z_j=0,z_{-j} \geq 0\}$ is a subset of $\mathcal{B}$.

    By definition of conditional density, $f_{Z_j|Z_{-j}}(z_j|z_{-j}\geq0) = \int_{z_{-j}\geq0} f_{Z_j|Z_{-j}}(z_j|z_{-j}) f_{-j}(z_{-j}) dz_{-j}$: the previous result implies the right-hand side to be continuous at $z_j=0$. This gives the following implication:
    \begin{gather}
    f_{Z_j|Z_{-j}}(z_j|z_{-j}\geq0) \text{ continuous at } z_j=0, \forall j.
    \end{gather}
\end{proof}

\subsection*{Proposition \ref{thm:cons}}

\begin{prop}
{\normalfont (Asymptotic distribution of $\hat{f}_{j,p}(z_j)$)}
Under Assumptions \ref{ass:cont}, \ref{ass:cdf} and \ref{ass:kernel}, with $p=2$, $n h_j^2 \rightarrow \infty$ and $n h_j^{2p+1}=O(1)$, $\hat{f}_{j,p}(z_j)$ is a consistent estimator for $f_{j}(z_j)$. Furthermore,
\begin{gather}
\sqrt{n_j h_j} (\hat{f}_{j,p}(z_j) - f_j(z_j) - h_j^p B(z_j)) \rightarrow^d \mathcal{N}(0,V(z_j))
\end{gather}
where $B(z_j)$ is the asymptotic bias.
\end{prop}

\begin{proof}
The sample size $n_j$ considered by the estimator $\hat{f}_{j,p}(z_j)$ is random. By the law of large numbers, $\frac{n_j}{n} = \frac{\sum_{i=1}^{n} \ind\{z_{-j}\geq 0\}}{n} \rightarrow^{a.s.} \pi_j >0$, and hence $n_j \rightarrow \infty$ with probability 1. With probability 1, then, the proposition is equivalent to the one stated and proved as Theorem 1 in \cite{cattaneo2020simple}.

Their results apply since $\hat{f}_{j,p}(z_j)$ can be written as
\begin{gather*}
\hat{f}_{j,p}(z_j) = e_1'\hat{\beta}(z_j) \\
\hat{\beta}(z_j) = \text{argmin}_{b\in \R^{p+1}} \sum_{i=1}^{n_j}\left[\tilde{F_j}(z_j) - r_p(z_{ji} -z_j)'b \right]^2 K\left(\frac{z_{ji} -z_j}{h_j}\right)
\end{gather*}
if only observations with $z_{-j} \geq 0$ are considered, with $n_j \rightarrow \infty$ with probability 1, $n_j h_j^2 \rightarrow \infty$ and $n_j h_j^{2p+1}=O(1)$.
\end{proof}

\subsection*{Proposition \ref{thm:snorm}}
\begin{prop}
{\normalfont (Asymptotic distribution of $\hat{\theta}_{j,q}$)}
Under Assumptions \ref{ass:cont}, \ref{ass:cdf} and \ref{ass:kernel} holding separately for $\{Z:Z\geq0\}$ and $\{Z:Z_{-j}\geq0,Z_j < 0\}$, with $p=2$, $q=p+1$, $n \min\{h_{j-},h_{j+}\} \rightarrow\infty$, and $n \max\{h_{j-}^{1+2q},h_{j+}^{1+2q}\} \rightarrow 0$, when the implication $\theta_{j}=0$ is true:
\begin{gather}
\frac{1}{\sigma_j} \hat \theta_{j,q} \rightarrow^d \mathcal{N}\left(0, 1 \right)
\end{gather}
where
\begin{gather}
\sigma_j^2 = \frac{\pi_{j+}}{h_{j+} \pi_j n} V_{j+}(0) + \frac{\pi_{j-}}{h_{j-} \pi_j n} V_{j-}(0).
\end{gather}
A consistent estimator $\hat \sigma_j^2$ for $\sigma_j^2$ can be obtained by 
\begin{gather}
\hat \sigma_j^2 =\frac{n_{j+}}{h_{j+} n_j^2} \hat V_{j+,q}(0) + \frac{n_{j-}}{h_{j-} n_j^2} \hat V_{j-,q}(0).
\end{gather}
\end{prop}

\begin{proof}
As for Proposition \ref{thm:cons}, the effective samples sizes $n_{j}$, $n_{j+}$ and $n_{j-}$ are stochastic. By the law of large numbers, $\frac{n_j}{n} \rightarrow^{a.s.} \pi_j>0$, $\frac{n_{j+}}{n_j} \rightarrow^{a.s.} \pi_{j+}>0$, and $\frac{n_{j-}}{n_j} \rightarrow^{a.s.} \pi_{j-}>0$, and hence $n_j \rightarrow \infty$, $n_{j+} \rightarrow \infty$, and $n_{j-} \rightarrow \infty$ with probability 1. Then, with probability 1, this proposition is equivalent to the result stated and proved as Corollary 1 in the appendix of \cite{cattaneo2020simple}.

Since $\hat V_{j+,q}(0) \to^p V_{j+}(0)$, $\hat V_{j-,q}(0) \to^p V_{j+}(0)$, $\frac{n_{j+}}{h_{j+} n_j^2} \to^p \frac{\pi_{j+}}{h_{j+} \pi_j n}$, and $\frac{n_{j-}}{h_{j-} n_j^2} \to^p \frac{\pi_{j-}}{h_{j-} \pi_j n}$, the Slutsky theorem implies $\hat \sigma_j^2 \to^p \sigma_j^2$.
\end{proof}

\subsection*{Theorem \ref{thm:theta}}
\begin{theorem} 
{\normalfont (Asymptotic distribution of $\hat{\theta}$)}
Under Assumptions \ref{ass:cont}, \ref{ass:cdf} and \ref{ass:kernel} holding separately for $\{Z:Z\geq0\}$ and $\{Z:Z_{-j}\geq0,Z_j < 0\}$ for all $j$, with $p=2$, $q=p+1$, $n \min\{h_{j-},h_{j+}\} \rightarrow\infty$ and $n \max\{h_{j-}^{1+2q},h_{j+}^{1+2q}\} \rightarrow 0$ for all $j$, when $\theta=0$,
\begin{gather*}
\hat{\Sigma}^{-\frac{1}{2}}\hat \theta \rightarrow^d \mathcal{N}(0,I).
\end{gather*}
where $\hat{\Sigma}_{jj} = \hat{\sigma}^2_j$ as defined in Proposition \ref{thm:snorm}, and $\hat{\Sigma}_{ji}=0$ for all $i \neq j$.
\end{theorem}

\begin{proof}

Proposition \ref{thm:snorm} derives the univariate asymptotic distribution of $\frac{1}{\sigma_j}\hat \theta_j$. Consider any pair of elements of $\hat{\Sigma}^{-\frac{1}{2}} \hat \theta$. Showing that they are asymptotically independent proves the theorem, as independent normal distributions are jointly normal.

Without loss of generality, consider $\frac{1}{\hat{\sigma_1}} \hat \theta_1$ and $\frac{1}{\hat{\sigma_2}} \hat \theta_2$, where:
\begin{gather} 
\hat\theta_{1,p} = \frac{n_{1+}}{n_1} \hat{f}_{1+,p}(0) - \frac{n_{1-}}{n_1} \hat{f}_{1-,p}(0) \\
\hat\theta_{2,p} = \frac{n_{2+}}{n_2} \hat{f}_{2+,p}(0) - \frac{n_{2-}}{n_2} \hat{f}_{2-,p}(0).
\end{gather}

In finite samples, $\hat\theta_{1,p}$ and $\hat\theta_{2,p}$ are not independent: $\hat f_{1+,p}(0)$ and $\hat f_{2+,p}(0)$ are computed with different but overlapping sets of observations. For an intuition, consider the bi-dimensional space of the two running variables: each estimator gives non-zero weights to observations in a stripe of width $h_{1+}$ or $h_{2+}$ close to the boundary. Close to the origin, the stripes overlap. $\hat f_{1+,p}(0)$ and $\hat f_{2+,p}(0)$ considers a number of observations proportional to $n_{1+} h_{1+}$ and $n_{1+} h_{2+}$, while the shared number is proportional to $n_{1+} h_{1+} h_{2+}$. Note that, under the assumptions on rates of convergence, $n_{1+} h_{1+} h_{2+} \rightarrow \infty$.

Write $\frac{1}{\hat{\sigma_1}} \hat \theta_1$ as:
\begin{gather}
\frac{1}{\hat{\sigma_1}} \hat \theta_1 = 
\left(\frac{1}{n_1 h_{1+}} \frac{n_{1+}}{n_1} \hat{V}_{1+,q}(0) +\frac{1}{n_1 h_{1-}}\frac{n_{1-}}{n_1} \hat{V}_{1-,q}(0) \right)^{-\frac{1}{2}} \hat \theta_1 = \\
\frac{n_{1}}{n_{1+}}\underbrace{\frac{n_{1+}}{n_1}\left(\frac{n_{1+}h_{1+}}{n_1 h_{1+}} \frac{n_{1+}}{n_1} \hat{V}_{1+,q}(0) +\frac{n_{1+}h_{1+}}{n_1 h_{1-}}\frac{n_{1-}}{n_1} \hat{V}_{1-,q}(0) \right)^{-\frac{1}{2}}}_{\hat{M}_{n_1}}(n_{1+}h_{1+})^{\frac{1}{2}}\hat \theta_1.
\end{gather}
To prove asymptotic independence of $\frac{n_{1}}{n_{1+}}\hat{M}_{n_1}(n_{1+}h_{1+})^{\frac{1}{2}}\hat \theta_1$ and $\frac{n_{2}}{n_{2+}}\hat{M}_{n_2}(n_{2+}h_{2+})^{\frac{1}{2}}\hat \theta_2$, I need to show that $\hat{M}_{n_1}(n_{1+}h_{1+})^{\frac{1}{2}}\hat f_{1+,p}(0)$ and $\hat{M}_{n_2}(n_{2+}h_{2+})^{\frac{1}{2}}\hat \hat f_{2+,p}(0)$ are independent.

Define $\hat f_{1+,p}^*(0)$ and $\hat{M}_{n_1}^*$ as the estimators analogous to $\hat f_{1+,p}(0)$ and $\hat{M}_{n_1}$ that consider only observations not in the overlapping region. I will show that $|\hat{M}_{n_1}(n_{1+}h_1)^\frac{1}{2}\hat f_{1+,p}(0) -\hat{M}_{n_1}^*(n_{1+}h_1)^\frac{1}{2}\hat f_{1+,p}^*(0)|\rightarrow 0$, and prove the theorem, since $\hat f_{1+,p}^*(0)$ and $\hat f_{2+,p}(0)$, and $\hat{M}_{n_1}^*$ and $\hat{M}_{n_2}$, are independent.

In the proof, where not necessary, I omit the subscripts $+$ and $p$, and the argument $z_j=0$, and consider vector $\hat \beta_j$, where $\hat f_{j+,p}^*(z_j) = e_1'\hat \beta_j (z_j)$. The local linear estimator can be written in the matrix form:
\begin{gather}
\hat \beta_j = H^{-1}(X'KX)^{-1}X'KY
\end{gather}
where $$H=\text{diag}\left(h_j^0,h_j^1,\dots,h_j^p\right)$$ $$X=\left[\left(\frac{z_{ji}}{h_j}\right)^m\right]_{1\leq i \leq n_{j}, 0 \leq m \leq p}$$ $$K=\text{diag}\left(\left[\frac{1}{h_j}K\left(\frac{z_{ji}}{h_j}\right)\right]_{1\leq i \leq n_{j+}} \right)$$
$$Y=\left(\tilde{F}_j(z_{ji})\right)_{1\leq i \leq n_{j}} $$
Matrices $X$ and $Y$ can be decomposed as $X=A+B$ and $Y=C+D$, where $A$ and $C$ have rows of zeros in correspondence with the overlapping observations. In contrast, $B$ and $D$ have rows of zero in correspondence of non-overlapping ones. Note that
\begin{gather}
\sqrt{n_{j}h_j} \hat \beta_j = \sqrt{n_{j}h_j} H^{-1}(A'KA + B'KB)^{-1}(A'KC + B'KD) = \\
H^{-1}\left(\underbrace{\frac{(n_{j}-n_jh_j)h_j}{n_jh_j}}_{\rightarrow1}
\underbrace{\frac{1}{(n_{j}-n_jh_j)h_j}A'KA}_{\rightarrow^p O(1)} + \underbrace{\frac{n_jh_j^2}{n_jh_j}}_{\rightarrow0}
\underbrace{\frac{1}{n_jh_j^2}B'KB}_{\rightarrow^p O(1)}\right)^{-1} \\
\left(\underbrace{\frac{\sqrt{(n_{j}-n_jh_j)h_j}}{\sqrt{n_jh_j}}}_{\rightarrow1}
\underbrace{\frac{1}{\sqrt{(n_{j}-n_jh_j)h_j}}A'KC}_{\rightarrow^p O(1)} + \underbrace{\frac{\sqrt{n_jh_j^2}}{\sqrt{n_jh_j}}}_{\rightarrow0}
\underbrace{\frac{1}{\sqrt{n_jh_j^2}}B'KD}_{\rightarrow^p O(1)}\right) \\
\rightarrow \sqrt{n_{j}h_j} H^{-1}(A'KA)^{-1}(A'KC) = \sqrt{n_{j}h_j} \hat \beta_j^*
\end{gather}

This demonstrates that $\left|\sqrt{n_{j}h_j} \hat \beta_j - \sqrt{n_{j}h_j} \hat \beta_j^*\right|\rightarrow 0$. Analogously, it can be shown that $\hat{M}_{n_1} \rightarrow M_1$ and $\hat{M}_{n_1}^* \rightarrow M_1$, with $M_1= \pi_{1+}\left(\pi_{1+}^2 V_{1+,q}(0) + \pi_{1+}\pi_{1-} \frac{h_{1+}}{h_{1-}} V_{1-,q}(0) \right)^{-\frac{1}{2}}$. The result is hence obtained: $$\left|\hat{M}_{n_1}(n_{j+}h_j)^\frac{1}{2}\hat f_{1+,p}(0) -\hat{M}_{n_1}^*(n_{j+}h_j)^\frac{1}{2}\hat f_{1+,p}^*(0)\right|\rightarrow 0.$$

Since $\hat \theta_j$ and $\hat \theta_i$ are asymptotically independent, $\hat \Sigma_{ji} = 0$ for all $j\neq i$.
\end{proof}

\subsection*{Corollary \ref{thm:coroll}}
\begin{cor}
{\normalfont (Manipulation test)}
Let $H_0$ be the null hypothesis reported in Equation \eqref{eq:condition_marginal}. Under the assumptions of Theorem \ref{thm:theta}, when $H_0$ is true:
\begin{gather}
\lim_{n\rightarrow \infty} P(\phi(\hat t,\alpha)=1) = \alpha.
\end{gather}
When $H_0$ is false:
\begin{gather}
\lim_{n\rightarrow \infty} P(\phi(\hat t,\alpha)=1) = 1.
\end{gather}
\end{cor}

\begin{proof}
The proof for the case of true null hypothesis immediately follows from Theorem \ref{thm:theta}: asymptotically, $\hat \Sigma^{-\frac{1}{2}}\hat \theta$ is distributed as a multinomial standard normal, and hence the quadratic form $\hat t$ is such that $\hat t = \hat \theta' \hat \Sigma^{-1} \hat \theta \rightarrow^d \chi^2_d$.

For the case when $H_0$ is false, note that it means that at least one of the conditional marginal densities is not continuous: it exists a $j$ such that $\theta_j \neq 0$, and hence $\frac{1}{\hat{\sigma}_j} \hat{\theta}_{j,q} \to \infty$. It implies $\hat{t} \to \infty$, and then $\lim_{n\rightarrow \infty} P(\phi(\hat t,\alpha)=1) = 1$.
\end{proof}

\subsection*{Theorem \ref{thm:ident}}
\begin{theorem} 
{\normalfont (Identification of $\tau(b)$ and $\tau$)}
Under Assumptions \ref{ass:cont} and \ref{ass:ycont}, $\tau(b)$ and $\tau$ are identified.
\end{theorem}
\begin{proof}
First, prove that $\E[Y_0|Z=z]$ and $\E[Y_1|Z=z]$ are continuous at $z=b\in\mathcal{B}$, where the expectation is taken with respect to random variables $W$, $Z$, and $y_0(w,z)$ and $y_1(w,z)$. Consider $\E[Y_0|Z=z]$:
\begin{align}
\E[Y_0|Z=z] =& \E[y_0(W,Z)|Z=z] = \int_{\mathcal{W}}\E[y_0(w,z)]f(w|z)dw = \\ =& \int_{\mathcal{W}}\E[y_0(w,z)]f(z|w)\frac{g(w)}{f(z)}dw = \\ =&
\int_{\mathcal{W}}\E[y_0(w,z)]f(z|w)\frac{g(w)}{\int_{\mathcal{W}} f(z|w)g(w)dw}dw.
\end{align}
By Assumption \ref{ass:cont}, $f(z|w)$ is continuous at $z=b\in\mathcal{B}$, and by Assumption \ref{ass:ycont}, $\E[y_0(w,z)]$ is continuous at $z=b\in\mathcal{B}$. Hence, $\E[Y_0|Z=z]$ is continuous at $z=b\in\mathcal{B}$. The proof for $\E[Y_1|Z=z]$ is analogous.

For any $b\in \mathcal{B}$, the Conditional Average Treatment Effect can be written as
\begin{align*}
\tau(b) =& \E[Y_1 -Y_0 |Z=b] = \E[Y_1 |Z=b] - \E[Y_0 |Z=b] = \\ =&
\lim_{z\rightarrow b} \E[Y_1|Z=z] - \lim_{z\rightarrow b} \E[Y_0|Z=z] = \\ =&
\lim_{z\rightarrow b, z \in \mathcal{T}} \E[Y|Z=z] - \lim_{z\rightarrow b, z \in \mathcal{T^C}} \E[Y|Z=z].
\end{align*}
In the last expression, the limits for $z \rightarrow b$ of $\E[Y|Z=z]$ are taken considering any sequence of $z \in \mathcal{T}$ and $z \in \mathcal{T^C}$. The last equality comes from the fact that, for sequences in $\mathcal{T}$, $Y=Y_1$, and for sequences in $\mathcal{T^C}$, $Y=Y_0$.

$Y$ and $Z$ are observable, and hence $\tau(b)$ is identified.

The Integrated Conditional Average Treatment Effect is then identified as
\begin{gather*}
\tau = \int_{b \in \mathcal{B}} \tau(b) f(b|Z\in \mathcal{B}) db = \frac{\int_{b \in \mathcal{B}} \tau(b) f(b) db}{\int_{b \in \mathcal{B}} f(b) db}.
\end{gather*}
\end{proof}

\end{document}